\documentclass[conference]{IEEEtran}

\usepackage{cite}
\ifCLASSINFOpdf
\else
\fi
\usepackage{xcolor}

\usepackage{amsfonts}
\usepackage{amssymb}
\usepackage{mathrsfs}
\usepackage{graphicx}
\usepackage{amsmath}
\usepackage{amsthm}
\RequirePackage{epsf}
\usepackage{epstopdf}
\newtheorem{theorem}{Theorem}
\newtheorem{corollary}{Corollary}
\newtheorem{lemma}{Lemma}

\begin{document}
\title{Design and Analysis of Resilient Vehicular \\ Platoon Systems over Wireless Networks}
\vspace{-0.2cm}
\author{\normalsize Tingyu Shui and Walid Saad \\ 
Bradley Department of Electrical and Computer Engineering, Virginia Tech, Arlington, VA, 22203, USA.
\\ Emails:\{tygrady, walids\}@vt.edu
\vspace{-0.5cm}}




\maketitle

\begin{abstract}
Connected vehicular platoons provide a promising solution to improve traffic efficiency and ensure road safety. Vehicles in a platoon utilize on-board sensors and wireless vehicle-to-vehicle (V2V) links to share traffic information for cooperative adaptive cruise control. To process real-time control and alert information, there is a need to ensure clock synchronization among the platoon's vehicles. However, adversaries can jeopardize the operation of the platoon by attacking the local clocks of vehicles, leading to clock offsets with the platoon's reference clock. In this paper, a novel framework is proposed for analyzing the resilience of vehicular platoons that are connected using V2V links. In particular, a resilient design based on a diffusion protocol is proposed to re-synchronize the attacked vehicle through wireless V2V links thereby mitigating the impact of variance of the transmission delay during recovery. Then, a novel metric named temporal conditional mean exceedance is defined and analyzed in order to characterize the resilience of the platoon. Subsequently, the conditions pertaining to the V2V links and recovery time needed for a resilient design are derived. Numerical results show that the proposed resilient design is feasible in face of a nine-fold increase in the variance of transmission delay compared to a baseline designed for reliability. Moreover, the proposed approach improves the reliability, defined as the probability of meeting a desired clock offset error requirement, by $45\%$ compared to the baseline. 
\end{abstract}

% \fi
\IEEEpeerreviewmaketitle

\vspace{-0.22cm}
\section{Introduction}
\vspace{-0.12cm}
Vehicular platoons are promising components of intelligent transportation systems (ITS), that can help improve traffic efficiency and road safety \cite{8667866}. By arranging independent vehicles into a platoon, vehicles are maintained with well-designed spacing and velocity, reducing fuel consumption due to less air resistance \cite{9699045}. By using communication link along the platoon, risk alerts can be transmitted to convey information on traffic issues, e.g., rear-end collision, platoon merging/splitting, and emergency braking \cite{7736181}. One critical challenge facing the deployment of connected vehicular platoon is the need for synchronization over all the vehicles to ensure a correct temporal ordering among different events \cite{9834918}.

% Considering the transition from traditional adaptive cruise control (ACC) relying on only on-board radar sensing to cooperative adaptive cruise control (CACC) enabling information sharing through wireless communication, the synchronization among vehicular platoon is necessary. 
Platoon synchronization requires all vehicles' local clocks to follow an identical platoon reference clock for processing real-time control and alert information. 
% \cite{mahmood2019time}. 
An inconsistency between a certain vehicle's local clock and platoon reference clock could have dire consequences such as inaccurate braking or acceleration \cite{1580935}. Although the vehicles of a platoon can be synchronized prior to deployment, cyber-attacks such as synchronization disruption \cite{s22176679} can still jeopardize the platoon's synchronization during its operation. Thus, one must develop a resilient platoon system that can recover from such attacks. An intuitive solution is to use information shared over wireless vehicle-to-vehicle (V2V) links. Since this information contains real-time timestamps, a compromised vehicle can directly use it to modify its clock. However, a significant drawback of this solution is that there exists a delay between the received timestamp and the actual time,
% \cite{6645401}, 
due to the stochastic wireless V2V transmission. Thus, an effective synchronization through V2V information sharing also requires the system to be resilient in face of the variance of the transmission delay.

% A natural solution is to rely upon the information transmitted along the platoon through wireless links. Since these information contains timestamps \cite{}, as a compromised solution, vehicles under cyber-attacks can directly modify their clock according to the timestamps.
To address this issue, a common approach is to design a reliable wireless V2V link that maintains a high probability of meeting a desired, low latency delay target. In this regard, recent works \cite{10213228,9046279,9741813,9013252} studied the low latency challenges of connected vehicular platoons. For example, the work in \cite{10213228} studied the problem of joint power control and spectrum allocation for an ultra-reliable platoon communication with a focus on designing a robust system under worst-case channel conditions. In \cite{9046279}, the authors modeled the uncertainty in platoon communication as a Markov decision process (MDP), and they used reinforcement learning to enable ultra reliable low latency V2V communications. Moreover, the work in \cite{9741813} considered a two time-scale resource allocation framework for ultra-reliable V2V communications, where both large-scale and instantaneous channel state information (CSI) are leveraged to reduce the occurrence of extreme events. Furthermore, the work in \cite{9013252} proposed a dependence control mechanism by modeling the correlation among co-existing V2V links delay and, then, used this correlation to optimize network reliability.

However, this prior art \cite{10213228,9046279,9741813,9013252} focuses on either a reliable or a robust design aiming to maintain a low transmission delay, but it ignores the system's response once the desired delay requirements are violated. In fact, since a vehicular platoon relies on real-time traffic information, the risks of delayed control or alert information will increase if the desired delay requirements are exceeded consecutively, over time. Moreover, reliable and robust designs aiming at a low outage probability will impose very strict requirements on the V2V links' signal-to-interference-plus-noise ratio (SINR), which limits the practicality of these two approaches. Therefore, instead of satisfying the V2V link delay requirement with high probability, there is a need to design a resilient system that can respond to unacceptable delays that violate the V2V link delay requirement and, then, recover from synchronization disruption attack.

The main contribution of this paper is a novel framework for designing resilient platoon systems that can recover from synchronization disruption attacks. Specifically, we propose a metric named temporal conditional mean exceedance (TCME) to characterize the resilience of the system. In particular, we adopt a diffusion protocol to mitigate the impact of the variance of the V2V link delay on the compromised vehicle's clock offset. We then define resilience as the platoon's ability to recover from synchronization disruption attacks and to prevent the clock offset from consecutively exceeding the desired threshold. To analyze the feasible region of our resilient design, we derive the cumulative distribution function (CDF) of the V2V link delay and an approximated expression of the TCME. \textit{To our best knowledge, this is the first work that defines the resilience of a connected vehicular platoon under synchronization cyber-attack and characterizes the conditions and time needed for recovery.} Simulation results show that the proposed resilient design achieves $45\%$ reliability gain and is feasible in face of a nine-fold increase in the variance of the transmission delay compared to a baseline reliable design that relies on outage probability.

The rest of the paper is organized as follows. Section II presents the system model. In Section III, we perform resilience analysis for the platoon. Section IV provides the simulation results, and conclusions are drawn in Section V.
\vspace{-0.10cm}
\section{System model}
\vspace{-0.15cm}
\subsection{Diffusion Strategy}
\vspace{-0.10cm}
We consider a single platoon within a two-dimensional space whose vehicles are organized into a predecessor-following (PF) model, in which each vehicle can receive information from its predecessor. However, synchronization disruption attacks can force vehicle $i$ to update its on-board software clock through the transmission of malicious message containing incorrect time information\cite{s22176679}. Hence, the local clock of vehicle $i$ will be incorrectly adjusted, causing a clock offset with the platoon reference clock, as shown in Fig. 1. To identify such an attack, self-check timestamps can be periodically transmitted together with cooperative adaptive cruise control (CACC) information along the platoon. When clock offset outliers are detected, vehicle $i$ will send a re-synchronization request to its predecessor $i-1$. It will also stop transmitting self-check timestamps to its follower $i+1$. Then, the re-synchronization process will start between $i-1$ and $i$ over wireless V2V links.

Due to the directional information flow in the platoon, i.e., from predecessor to follower, synchronization protocols involving bidirectional communication are difficult to implement. As such, we can use the so-called diffusion strategy \cite{1566581} in which multiple rounds of communications from $i-1$ to $i$ are utilized to mitigate the synchronization error stemmed from variance of V2V link transmission delay. Specifically, at the beginning of time slot $l$, vehicle $i-1$ transmits clock information to vehicle $i$ using a V2V link. Then, vehicle $i$ records the received clock information. At the beginning of time slot $l+1$, vehicle $i$ modifies its clock reading according to the following diffusion update rule:
\vspace{-3pt}
\begin{equation}
\label{Clock}
C_i^{l+1} = \theta C_i^{l} + (1-\theta) \left[ C^{l} - \left( \tau_{i-1,i}^{l} - \mu_i \right) \right] + T,
\end{equation}
where $C_i^{l+1}$ is the modified clock reading of vehicle $i$ at the beginning of time slot $l+1$, $C_i^{l}$ is the clock reading at the beginning of time slot $l$, $C^{l}$ is the clock reading transmitted from vehicle $i-1$ at the beginning of time slot $l$ (also the platoon's reference clock reading), $T$ is the fixed interval of a single time slot such that $C^{l+1} = C^{l} + T$, and $\theta \in (0, 1)$ is the diffusion factor adopted by vehicle $i$. Note that the V2V link transmission delay $\tau_{i-1, i}^{l}$ from $i-1$ to $i$ is captured in \eqref{Clock}, which will lead to a delay in the modified clock. Thus, vehicle $i$ will mitigate the impact of the transmission delay by compensating for its expectation $\mu_i  = \mathbb{E} \left[\tau_{i-1,i}^{l} \right] $ when modifying its local clock\footnote{ Note that the assumption that vehicle $i$ is aware of $\mu_i$ is reasonable, as we will show in Section III.}. Accordingly, the clock offset $\xi_i^{l+1} = C_{i}^{l+1} - C^{l+1}$ at time slot $l+1$ between vehicle $i$ and the platoon reference clock will be:
    \vspace{-2pt}
\begin{equation}
\label{Error}
\xi_i^{l+1} = \theta \xi_i^{l} - (1-\theta) u_{i-1,i}^{l} , l \in \mathbb{N}_{\geq 0},
\end{equation}
where $u_{i-1,i}^{l} = \tau_{i-1,i}^{l}  - \mu_i $ is the compensated delay. We assume that the initial clock offset $\xi_{0} = C_{i}^{0} - C^{0}$ follows a normal distribution, i.e., $\xi_{0} \sim \mathcal{N}(0,\sigma_0^2)$.

\subsection{Wireless Communication Model}
\begin{figure}[t]
\label{System Model}
	\centering
	\includegraphics[scale=0.38]{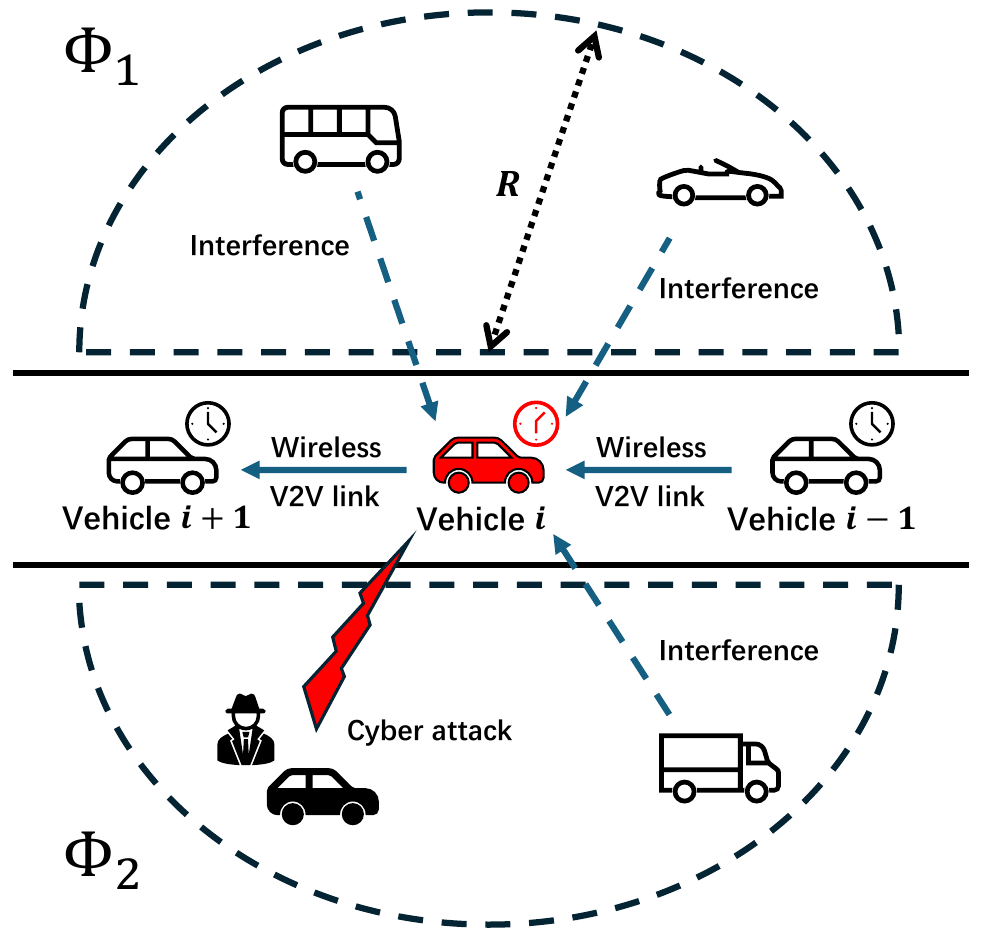}
    \vspace{-5pt}
	\caption{\small{A PF model where vehicles can only receive information from its predecessor. After cyber-attack, vehicle $i$ is working with clock offset to the platoon reference clock. }}
 \vspace{-0.5cm}
\end{figure}
Similar to \cite{8778746}, an orthogonal frequency-division multiple access (OFDMA) scheme is adopted within the platoon. 
% Thus, the V2V links between $i-1$ and $i$ will not experience interference from other vehicles in the platoon. 
Moreover, the considered system includes vehicles other than the platoon who will reuse the entire bandwidth. Those vehicles are distributed according to a Poisson point process (PPP) with intensity $\eta$ in two semicircle areas $\Phi_1$ and $\Phi_2$ of radius $R$ separated by the platoon and centered at vehicle $i$, as shown in Fig. 1. We consider a Nakagami channel for the V2V link between $i-1$ and $i$. Thus, at time slot $l$, the received power at $i$ will be defined as $P_{i-1,i}^{l} = P^t_{i-1} g_{i-1,i}^{l} d_{i-1,i}^{-\alpha}$ where $P^t_{i-1}$ represents the transmit power of vehicle $i-1$ , $g_{i-1,i}^{l}$ is the channel gain following a Gamma distribution with shape parameter $m$, $d_{i-1,i}$ is the distance between $i$ and $i-1$, and $\alpha$ is the path loss exponent. Because of the lack of continuous line-of-sight links between vehicle $i$ and vehicles outside the platoon, we consider Rayleigh channels for the interference links to $i$. The interference vehicle $i$ experiences at time slot $l$ can be modeled as $I_{i,\Phi_1}^{l} = \sum_{k \in \Phi_1} P^t_k g_{k,i}^{l} (d_{k,i}^l)^{^{-\alpha}}$ and $I_{i,\Phi_2}^{l} = \sum_{k \in \Phi_2} P^t_k g_{k,i}^{l} (d_{k,i}^l)^{^{-\alpha}}$ where $P^t_k$ is the transmit power of vehicle $k$, $g_{k,i}^{l}$ is the channel gain, and $d^l_{k,i}$ is the distance between $k$ and $i$. Thus, the SINR of the V2V link between vehicle $i$ and its predecessor $i-1$ will be:
\begin{equation}
\label{Interference}
    \gamma_{i-1,i}^{l}=
    \frac{ P_{i-1,i}^{l} }{ I_{i,\Phi_1}^{l} + I_{i,\Phi_2}^{l}  + N_0B },
    \vspace{-2pt}
\end{equation}
where $B$ is the allocated bandwidth of the wireless channel and $N_0$ is the noise power spectral density. The transmission delay over the V2V link at time slot $l$ is modeled as $\tau_{i-1,i}^{l} = \frac{D}{B \log( 1 + \gamma_{i-1,i}^{l}) }$ with $D$ being  the V2V link packet size.
% Such consideration on two adjacent clock offsets $\xi_i^{t}$ and $\xi_i^{t+1}$ is valuable in CACC, because the delay in the cooperative control information transmission of two vehicles will greatly impact the stability of platoon \cite{}. If the clock of vehicle $i$ is biased, i.e., $\xi_i^{t} \neq 0$, such delay contains more uncertainty as we have to simultaneously consider the V2V link transmission delay and clock offset under cyber-attack. Thus, if we can restrain the clock offsets so that they will not exceed a threshold within two adjacent time slots, then a control strategy designed with safety margin is able to keep the platoon operational even vehicle $i$ is under cyber-attack. In the next section, we will first introduce a novel resilience metric to quantify the performance of diffusion strategy in \eqref{Error} and then present a theorem on the specific resilience to clock offset that vehicle $i$ can possess.
\subsection{Risk of Collision Measurement}
To characterize the impact of the clock offset $\xi_i^{l}$, we introduce the time to collision (TTC) metric to capture the rear-end collision risk \cite{546270}. At any time $t$, the relative distance $X_i(t)$ between vehicle $i$ and $i-1$ can be written as: 
\begin{equation}
\label{distance}
\begin{aligned}
    X_i(t)= & X_i\left(t_0\right)+\left[V_{i-1}\left(t_0\right)-V_{i}\left(t_0\right)\right]\left(t-t_0\right) \\ 
    &+\int_{t_0}^t \int_{t_0}^u\left[a_{i-1}(s)-a_i(s)\right] \mathrm{d} s \mathrm{d} u,
\end{aligned}
    \vspace{-2pt}
\end{equation}
where $V_{i}$ and $V_{i-1}$ represent the velocities of vehicle $i$ and its predecessor $i-1$, $a_{i}$ and $a_{i-1}$ represent their accelerations, respectively, and $t_0$ is the time when the measurement of rear-end collision risk is required. To measure the rear-end collision risk, TTC is defined as the interval after which vehicle $i$ will collide with vehicle $i-1$, i.e., $T_i^c(t_0) = \inf \left\{t^* - t_0 \mid X_i(t^*) \leq 0 \right\}$. Accordingly, vehicle $i$ can potentially experience a rear-end collision if $T_i^c(t_0) < \hat{t}$ \cite{LI2020105676}, where $\hat{t}$ is a minimum interval for manual and autonomous control to avoid collision.

% \begin{figure}[t]
% \label{TTC}
% 	\centering
% 	\includegraphics[scale=0.28]{ttc.pdf}
%     \vspace{-4pt}
% 	\caption{\small{Deceleration process in two scenarios: (a) Vehicle $i$ is synchronized with $\xi_i^{t_0} = 0$}; (b) Vehicle $i$ operates with clock offset $\xi_i^{t_0} < 0$.}
%      \vspace{-0.5cm}
% \end{figure}
Here, we consider the TTC of vehicle $i$ during deceleration. Specifically, we assume that the attack happens at time slot $l = 0$ and the re-synchronization of vehicle $i$ begins immediately. After the diffusion strategy proceeds for a while, vehicle $i$ receives braking alert information from its predecessor in time slot $l_0$, and it begins to decelerate. We define the beginning of time slot $l_0$ as $t_0$, then, the TTC measured at $t_0$ is $T_i^c(t_0)$, at which the clock offset of vehicle $i$ is $\xi_i^{l_0}$. For tractability, along the lines of the model in \cite{546270}, we consider a simple case with fairly constant braking capabilities, i.e., $a_{i-1} = a_i = a \leq 0$. Moreover, we assume a constant and homogeneous speed and relative distance $V_{i-1}(t_0) = V_{i}(t_0) = V, X_i(t_0) = d_{i-1,i} = X$. The impact by clock offset $\xi_i^{l_0}$ is considered in the deceleration delay after vehicle $i$ receives alert information. Particularly, we assume a constant processing delay $t_d$ and clock offset $\xi_i^{l_0} < 0 $ \footnote{The case $\xi_i^{l_0} > 0$ is equivalent to $\xi_i^{l_0} < 0$ by considering the rear-end collision between vehicle $i$ and its follower $i+1$.}, then the overall deceleration delay is given as $\Delta t = t_d - \xi_i^{l_0}$. From a safety perspective, we consider the worst case $X = t_d V $. Specifically, collision will occur if there exists a non-zero clock offset with the platoon's reference clock. Thus, the relative distance in \eqref{distance} can be derived as:
\begin{equation}
\vspace{-0.1cm}
\label{casedistance}
\begin{aligned}
&X_i(t) = X_i\left(t_0\right)\\
+&
\begin{cases}
 \frac{a}{2} (t - t_0)^2, \ & t_0 \leq t < t_0 + \Delta t,\\
 a \Delta t (t - t_0) - \frac{a}{2} {\Delta t}^2, \ & t_0 + \Delta t  \leq t < t_m,\\
\begin{matrix}
a \Delta t (t_m - t_0) - \frac{a}{2} {\Delta t}^2 + \\
a \Delta t (t - t_m)  - \frac{a}{2}  (t - t_m)^2 \\
\end{matrix}, \ & t_m  \leq t \leq t_m + \Delta t,
 % a \Delta t (t_m - t_0) - \frac{a}{2} {\Delta t}^2 +  a \Delta t (t - t_m)  - \frac{a}{2}  (t - t_m)^2, \quad & t_m  \leq t \leq t_m + \Delta t, \\
\end{cases}
\end{aligned}
\vspace{-0.1cm}
\end{equation}
where $t_m = t_0 - \frac{V}{a}$ is the time vehicle $i-1$ stops. From \eqref{casedistance}, we can obtain $T_i^c(\xi_i^{t_0})$ as follows:
% \begin{equation}
% \begin{aligned}
% &\operatorname{TTC}(t_0) \\
% =&
% \begin{cases}
% - \frac{V}{a} + \Delta t - \sqrt{-2 \frac{V}{a}(\Delta t - t_d)} \quad & t_d \leq \Delta t  < t_1 \\
% \frac{\Delta t}{2} + \frac{X}{-a \Delta t} \quad & t_1 \leq \Delta t  < t_2 \\
% \sqrt{\frac{2X}{-a}} \quad & t_2 \leq \Delta t \\
% \end{cases}
% \end{aligned}
% \end{equation}
% where $t_1 = - \frac{V}{a} - t_d - \sqrt{(\frac{V}{a})^2 + \frac{2X}{a}}$ and $t_2 = \sqrt{\frac{2X}{-a}} - t_d$.
\begin{equation}
\vspace{-0.1cm}
\label{casettc}
\begin{aligned}
T_i^c(\xi_i^{l_0}) = 
\begin{cases}
t_d - \xi_i^{l_0} - \frac{V}{a} - \frac{\sqrt{2aV\xi_i^{l_0}}}{a}, & t_1 \leq \xi_i^{l_0}  < 0 , \\
\frac{t_d - \xi_i^{l_0}}{2} - \frac{X}{a (t_d - \xi_i^{l_0})}, & t_2 \leq  \xi_i^{l_0}   < t_1, \\
\sqrt{\frac{2X}{-a}}, & \xi_i^{l_0} \leq t_2,
\end{cases}
\end{aligned}
\vspace{-0.1cm}
\end{equation}
where $t_1 =  t_d + \frac{V - \sqrt{V^2 + 2aX }}{a}$ and $t_2 = t_d - \sqrt{\frac{2X}{-a}} $. Here, we rewrite $T_i^c(t_0)$ as $T_i^c(\xi_i^{l_0})$ since $t_0$ is the beginning of time slot $l_0$. Moreover, $T_i^c(\xi_i^{l_0})$ is an increasing function of $\xi_i^{l_0}$, and, thus, to maintain the TTC above a threshold $\hat{t}$, we have $T_i^c(\xi_i^{l_0}) \geq \hat{t} \Leftrightarrow -\hat{\epsilon} \leq \xi_i^{l_0} \leq 0 $. Since the risk of collision from both the follower and predecessor of vehicle $i$ should be considered, the rear-end collision can be avoided if $(\xi_i^{l_0})^2 \leq \hat{\epsilon}^2$. For notational simplicity, we replace $l_0$ with $l$ as $l_0$ could be an arbitrary time slot.

To ensure $(\xi_i^{l})^2 \leq \hat{\epsilon}^2$, a ``robust design" focuses on the worst-case conditions and a ``reliable design" aims at a low outage probability. However, those two approaches do not consider the response to violations of $(\xi_i^{l})^2 \leq \hat{\epsilon}^2$. When the wireless delay changes significantly, clock offset may exceed the threshold in two consecutive time slots. Consequently, the control signals or risk alerts will experience exacerbated delay, which increases the collision risk. Instead, one must look at a \emph{resilient design} of the platoon, whereby we utilize the temporal relation of $\xi_i^{l}$ in \eqref{Error} to prevent two consecutive violations of the clock offset requirement at time slot $l$ and $l+1$, and enable recovery from synchronization disruption attack.

\vspace{-0.10cm}
\section{Resilience Analysis}
% In this section, we first analyze the evolution of $\xi_i^{l}$ according to \eqref{Error} and the difficulty in meeting $(\xi_i^{l})^2 \leq \hat{\epsilon}^2$. Then, we introduce TCME as a novel metric to capture system resilience and, then, we derive the conditions for resilient design as well as the recovery time.
\subsection{Clock Offset Analysis}
As shown in \eqref{Error}, $ \xi_i^{l+1}$ is composed of clock offset $ \xi_i^{l}$ and compensated delay $u_{i-1,i}^l$ in the last time slot $l$. Since the diffusion update rule \eqref{Clock} assumes that vehicle $i$ is aware of the expectation on transmission delay $\tau_{i-1,i}^{l}$ from its predecessor. Thus, we first derive the numerical expression of
$\mu_i = \mathbb{E} \left[\tau_{i-1, i}^{l} \right]$ and $\sigma_i^2 = \mathbb{V}\left[ \tau_{i-1, i}^{l} \right]$.
\begin{lemma}
The expectation and variance of the transmission delay $\tau_{i}^{l}$ of the V2V link between vehicle $i-1$ and vehicle $i$ is given by:
\vspace{-0.12cm}
\begin{equation}
\label{mu}
    \mu_i = \mathbb{E} \left[\tau_{i-1, i}^{l} \right] = \int_{0}^{\infty} \frac{D}{B \log_2(1 + \gamma)} f_{i}(\gamma) \mathrm{d} \gamma,
\end{equation}
\begin{equation}
\label{var}
    \sigma_i^2 = \mathbb{V}\left[ \tau_{i-1,i}^{l} \right] =  \int_{0}^{\infty} \left( \frac{D}{B \log_2(1 + \gamma)}\right)^2 f_{i}(\gamma) \mathrm{d} \gamma - \mu_i^2,
\end{equation}
where $f_{i}(\gamma) = \frac{\mathrm{d} F_{i}(\gamma)}{\mathrm{d}\gamma}$ and $ F_{i}(\gamma) = 1 - \sum_{k=1}^m (-1)^{k+1} {\binom{n}{m}} \operatorname{exp} \left( - k \eta \gamma \frac{d_{i-1,i}^{\alpha}}{P_{i-1} } B N_0 \right) \mathcal{L}^2_{I_{i,\Phi_1}^{l}} \left( k \eta \gamma \frac{d_{i-1,i}^{\alpha}}{P_{i-1} }  \right)$ with $\mathcal{L}_{I_{i,\Phi_1}^{l}} \left(s\right)$ being the Laplace transform of $I_{i,\Phi_1}^{l}$. 
\end{lemma}
\begin{proof}
The proof was omitted due to space limitation.
\end{proof}
From Lemma 1, we can observe that $\mathbb{E} \left[ u_{i-1,i}^{l} \right] = 0$ and $\mathbb{V} \left[ u_{i-1,i}^{l} \right] = \sigma_i^2$. By replacing the term $\xi_i^{l}$ through recursive formulation, we can rewrite \eqref{Error} as follows:
\vspace{-0.2cm}
\begin{equation}
\label{recursive}
    \xi_i^{l+1} = \theta^{l+1} \xi_{i}^{0} -  (1-\theta) \sum_{k=0}^{l} \theta^k u_{i-1,i}^{l-k}, l \in \mathbb{N}_{\geq 0}.
    \vspace{-0.08cm}
\end{equation}
From \eqref{recursive}, the impact of the initial clock offset $\xi_i^0$ will be eliminated as the diffusion protocol proceeds. Eventually, the clock offset only contains the scaled compensated delay. However, it is difficult to ensure a high probability to maintain $(\xi_i^{l})^2 \leq \hat{\epsilon}^2$. First, the diffusion protocol needs multiple rounds to mitigate the impact of the initial clock offset $\xi_i^{0}$, as per the term $\theta^{l+1} \xi_{i}^{0}$. Thus, the requirement is particularly stringent for small $l$ at the beginning of recovery. Although we can adopt a small $\theta$ to mitigate the impact of $\xi_i^{0}$ rapidly, the variance of $(1-\theta) u_{i-1,i}^{l}$ will increase as shown in \eqref{Error}. In other words, if we have already maintained $(\xi_i^{l})^2 \leq \hat{\epsilon}^2$ at time slot $l$, the probability of violating this requirement in the next time slot $l+1$ will increase when $\theta$ is small. 
\addtolength{\topmargin}{0.1in}
Due to the difficulty in maintaining $(\xi_i^{l})^2 \leq \hat{\epsilon}^2$ with high probability, next, we consider a resilience-focused approach that prevents violations of the clock offset requirement in two consecutive time slot. 
 
% Specifically, the clock offset is expected to recover from high-risk situation where $(\xi_i^{l})^2 \geq \hat{\epsilon}^2$ and avoid $(\xi_i^{l+1})^2 \geq \hat{\epsilon}^2$ in the next time slot $l+1$. 

% In contrast to a robust or reliable system that maintains the clock offset below a threshold, we focus on a resilient system that recovers from the situation in which clock offset exceeds a threshold. 

% However, it is challenging to derive the distribution of $\xi_i^{t}$ as an analytic distribution of $u_i^{t}$ is hard to obtain. For technical tractability, we assume that $u_i^{t}$ tends to a normal distribution, which is a realistic approximation when modeling the distribution of transmission delay \cite{}. Then we have $\xi_i^{t} \in \mathcal{N}(0,\sigma_t^2)$, where $\sigma_t^2 = \theta^{2t} \sigma_0^2 + \frac{1-\theta}{1+\theta} (1-\theta^{2t}) \sigma_i^2$ due to the independence of $u_i^{t}$ among any two time slots. 

\subsection{Temporal Conditional Mean Exceedance}
To illustrate the resilience of our system and investigate factors impacting resilience, first, we propose a new metric inspired by conditional mean exceedance (CME) \cite{bryson1974heavy}. Specifically, we define the TCME metric as follows:
\vspace{-0.12cm}
\begin{equation}
\label{TCME}
    h^l(\epsilon) = \mathbb{E} \left\{ (\xi_i^{l+1})^2 - \epsilon \mid (\xi_i^{l})^2 > \epsilon  \right\},
    \vspace{-0.08cm}
\end{equation}
where $\xi_i^{l+1}$ and $\xi_i^{l}$ are temporally correlated according to \eqref{Error}. To instill resilience into the system, if the clock offset of vehicle $i$ violates the requirement $(\xi_i^{l})^2 \leq \hat{\epsilon}^2$ at time slot $l$, then, the system must avoid $(\xi_i^{l+1})^2 \geq \hat{\epsilon}^2$ in the next time slot $l + 1$ to reduce the risk of collision. Thus, in our system, resilience is equivalent to guaranteeing $h^l(\hat{\epsilon}^2) < 0$. This condition means that, the clock offset will respond to the violation on $(\xi_i^{l})^2 \leq \hat{\epsilon}^2$ and recover from it after one round of diffusion, i.e., $\mathbb{E} \left\{ (\xi_i^{l+1})^2 \mid (\xi_i^{l})^2 > \hat{\epsilon}^2  \right\} \leq \hat{\epsilon}^2$. Note that, following the attack at time slot $0$, the diffusion protocol will proceed for a certain interval to guarantee $h^l(\hat{\epsilon}^2) < 0$ at time slot $l$. However, prior to time slot $l$, the platoon will still experience a high risk of collision. Thus, $l$ can capture the recovery time. As such, a small $l$ is desired in order to reduce the interval of high risk and enhance the resilience of the platoon. We can now ask two fundamental questions related to resilience: 1) What is the condition on the variance of V2V link delay that is needed to satisfy $h^l(\hat{\epsilon}^2) < 0$? 2) What is the minimum $l$ for satisfying $h^l(\hat{\epsilon}^2) < 0$? 
\addtolength{\topmargin}{-0.08in} 

To answer these two questions, we must obtain the closed-form distribution of $\xi_i^{l} = \theta^{l} \xi_{i}^{0} -  (1-\theta) \sum_{k=0}^{l-1} \theta^k u_{i-1,i}^{l-1-k}$, which, however, is challenging because of the expressions of $f_i(\gamma)$ and $F_i(\gamma)$ given in Lemma 1. Thus, for tractability, we approximate $Y_l = \sum_{k=0}^{l-1} \theta^k u_{i-1,i}^{l-1-k}$ to a normal distribution. Such approximation is reasonable as we can derive $\sup _{y \in \mathbb{R}}\left|G_l(y)-\Phi_l(y)\right| \leq C \cdot \frac{\sqrt{1-\theta^2}^{3}}{1-\theta^3} \cdot \frac{1-\theta^{3l}}{\sqrt{1-\theta^{2l}}^3}$ by using the Berry–Esseen theorem \cite{berry1941accuracy}, which shows a bounded error of the approximation. Specifically, $\Phi_l(y)$ is the CDF of $Y_n \sim \mathcal{N}(0,\sigma_l^2)$ with $\sigma_l^2 = \theta^{2l} \sigma_0^2 + \frac{1-\theta}{1+\theta} (1-\theta^{2l}) \sigma_i^2$, $G_l(y)$ is the CDF of $Y_l$ and $C$ is a constant. Since the bound is a decreasing function of $\theta$ and converges to $0$ when $\theta$ converges to $1$, the error of the approximation will be small if we select a large $\theta$. Next, we derive the expression of $h^l(\epsilon)$ based on the approximation $\xi_i^{l} \sim \mathcal{N}(0,\sigma_l^2)$.
\begin{lemma}
\vspace{-0.1cm}
Given that $\xi_i^{l} \sim \mathcal{N}(0,\sigma_l^2)$, at time slot $l$, we have:
    \vspace{-0.2cm}
\begin{equation}
\label{lemma1}
    h^l(\epsilon) = \theta^2 \sigma_l^2 +  \theta^2 \sqrt{\frac{2}{\pi}}\sigma_l \frac{ \sqrt{\epsilon} \operatorname{exp}  (-\frac{\epsilon}{2\sigma_l^2})  }{1-\operatorname{erf}(\frac{\sqrt{\epsilon}}{\sqrt{2}\sigma_l})} + (1-\theta)^2 \sigma_i^2 - \epsilon,
    \vspace{-0.12cm}
\end{equation}
where $\operatorname{erf}(x) =\frac{2}{\sqrt{\pi}} \int_0^z e^{-t^2} \mathrm{d} t$ is the Gaussian error function.
\end{lemma}

\begin{proof}
Based on \eqref{TCME}, we have $h^l(\epsilon) = \mathbb{E} \left\{ (\xi_i^{l+1})^2 - \epsilon \mid (\xi_i^{l})^2 > \epsilon  \right\} = \theta^2 \mathbb{E} \left\{ (\xi_i^{l})^2 \mid (\xi_i^{l})^2 > \epsilon  \right\} + (1-\theta)^2\sigma_i^2 - \epsilon$ with $\xi_i^{l} \sim \mathcal{N}(0,\sigma_l^2)$. Then, we derive the conditional expectation $\mathbb{E} \left\{ (\xi_i^{l})^2 \mid (\xi_i^{l})^2 > \epsilon  \right\}$ as 
    \vspace{-0.15cm}
\begin{equation}
    \mathbb{E} \left\{ (\xi_i^{l})^2 \mid (\xi_i^{l})^2 > \epsilon  \right\} 
    =   \frac{\frac{2}{\sqrt{2\pi}\sigma_l}\int_{\sqrt{\epsilon}}^{\infty} x^2 \operatorname{exp}(-\frac{x^2}{2\sigma_t^2}) \mathrm{d}x}{ \frac{2}{\sqrt{2\pi}\sigma_l} \int_{\sqrt{\epsilon}}^{\infty} \operatorname{exp}(-\frac{x^2}{2\sigma_l^2}) \mathrm{d}x},
    \vspace{-0.07cm}
\end{equation}
where we have $\frac{2}{\sqrt{2\pi}\sigma_l}\int_{\sqrt{\epsilon}}^{\infty} \operatorname{exp}(-\frac{x^2}{2\sigma_l^2}) \mathrm{d}x = 1 - \operatorname{erf}(\frac{\sqrt{\epsilon}}{\sqrt{2}\sigma_l})$ by definition. We can further derive the numerator as $\frac{2}{\sqrt{2\pi}\sigma_l}\int_{\sqrt{\epsilon}}^{\infty} x^2 \operatorname{exp}(-\frac{x^2}{2\sigma_l^2}) \mathrm{d}x = \sigma_l^2 \left[1 - \operatorname{erf}( \frac{ \sqrt{\epsilon} }{ \sqrt{2} \sigma_l } )\right] + \sqrt{\frac{2}{\pi}} \sigma_l \sqrt{\epsilon} \operatorname{exp} ( - \frac{\epsilon}{2\sigma_l^2})$.
% We can further express the numerator as follows that 
% \begin{equation}
% \begin{aligned}
% &\frac{2}{\sqrt{2\pi}\sigma_l}\int_{\sqrt{\epsilon}}^{\infty} x^2 \operatorname{exp}(-\frac{x^2}{2\sigma_l^2}) \mathrm{d}x \\
% = & 2 \sqrt{\frac{a}{\pi}} \int_{\sqrt{\epsilon}}^{\infty} x^2  \operatorname{exp}(-ax^2) \mathrm{d}x \\
% % = & 2 \sqrt{\frac{a}{\pi}} \left( - \frac{\mathrm{d}}{\mathrm{d}a} \int_{\sqrt{\epsilon}}^{\infty} \operatorname{exp}(-ax^2) \mathrm{d}x\right) \\
% = & 2 \sqrt{\frac{a}{\pi}} \left( - \frac{\mathrm{d}}{\mathrm{d}a} \int_{\sqrt{\epsilon a}}^{\infty} \frac{\operatorname{exp}(-t^2)}{\sqrt{a}} \mathrm{d}t\right) \\
% % = &  \sqrt{a} \left[ - \frac{\mathrm{d}}{\mathrm{d}a} \left(\frac{1 - \operatorname{erf}(\sqrt{\epsilon a})}{\sqrt{a}}\right) \right] \\
% = & \sigma_l^2 \left[1 - \operatorname{erf}( \frac{ \sqrt{\epsilon} }{ \sqrt{2} \sigma_l } )\right] + \sqrt{\frac{2}{\pi}} \sigma_l \sqrt{\epsilon} \operatorname{exp} ( - \frac{\epsilon}{2\sigma_l^2}). 
% \end{aligned}
% \end{equation}
By bringing both the numerator and denominator, we obtain the result in \eqref{lemma1}.
% \begin{equation}
% \begin{aligned}
%      h^l(\epsilon) 
%      & = \theta^2 \mathbb{E} \left\{ (\xi_i^{l})^2 \mid (\xi_i^{l})^2 > \epsilon  \right\} + (1-\theta)^2\sigma_i^2 - \epsilon \\
%      & = \theta^2 \sigma_l^2 +  \theta^2 \sqrt{\frac{2}{\pi}}\sigma_t \frac{ \sqrt{\epsilon} \operatorname{exp}  (-\frac{\epsilon}{2\sigma_l^2})  }{1-\operatorname{erf}(\frac{\sqrt{\epsilon}}{\sqrt{2}\sigma_l})} + (1-\theta)^2 \sigma_i^2 - \epsilon.
% \end{aligned}
% \end{equation}
\end{proof}
\vspace{-0.2cm}
Using the result of Lemma 2, we can derive key conditions on $\sigma_i^2$ to satisfy $h^l(\hat{\epsilon}^2) < 0$,
\begin{theorem}
Given a clock offset error threshold $\hat{\epsilon}$ and a desired time slot $l$, $h^l(\hat{\epsilon}^2) < 0$ can be satisfied in the platoon only when there exists $\theta \in (0,1)$ satisfying
    \vspace{-0.14cm}
\begin{equation}
\label{feasibility}
    \hat{\epsilon}^2 - (1-\theta)^2\sigma_i^2 \geq \frac{\theta^2\left( \hat{\epsilon} + \sqrt{\hat{\epsilon}+ 4\sigma_l^2 } \right)^2}{4}.
    \vspace{-0.1cm}
\end{equation}
\end{theorem}
\vspace{-0.2cm}
\begin{proof}
To satisfy $h^l(\hat{\epsilon}^2) < 0$, we first find a tight upper bound \cite{itō1974diffusion} of $h^l(\epsilon)$ as following
\vspace{-0.25cm}
\begin{equation}
\label{bound}
    \begin{aligned}
        h^l(\epsilon) 
        % & = \theta^2 \sigma_t^2 +  \theta^2 \sqrt{\frac{2}{\pi}}\sigma_l \frac{ \sqrt{\epsilon} \operatorname{exp}  (-\frac{\epsilon}{2\sigma_l^2})  }{1-\operatorname{erf}(\frac{\sqrt{\epsilon}}{\sqrt{2}\sigma_l})} + (1-\theta)^2 \sigma_i^2 - \epsilon \\
        & \leq \theta^2 \left( \sigma_l^2 + \frac{\epsilon}{2} + \sqrt{\frac{\epsilon^2}{4} + \epsilon \sigma_l^2 } \right) + (1 - \theta)^2 \sigma_i^2 - \epsilon. 
    \end{aligned}
\end{equation}
Consider feasibility of $ \theta^2 \left( \sigma_l^2 + \frac{\hat{\epsilon}^2}{2} + \hat{\epsilon} \sqrt{\frac{\hat{\epsilon}^2}{4} + \sigma_l^2 } \right) + (1 - \theta)^2 \sigma_i^2 - \hat{\epsilon}^2 \leq 0$, under condition $\hat{\epsilon}^2 - (1-\theta)^2\sigma_i^2 \geq \frac{\hat{\epsilon}^2}{4} \theta^2 $, we further transform it as
\begin{equation}
\label{sigmat}
    \sigma_l^2 \leq \left( \frac{\sqrt{\hat{\epsilon}^2 -(1-\theta)^2 \sigma_i^2 }}{\theta} -\frac{\hat{\epsilon}}{2} \right)^2 -\frac{\hat{\epsilon}^2}{4}.
\end{equation}
% Note that \eqref{sigmat} is equivalent to $\hat{\epsilon}^2 - (1-\theta)^2\sigma_i^2 \geq M_2(\theta) = \frac{\left(\hat{\epsilon} + \theta \sqrt{\hat{\epsilon}^2 + 4 \sigma_l^2}\right)^2}{4}$, thus, $h^l(\hat{\epsilon}^2) < 0$ can be satisfied if we have result in \eqref{feasibility}. 
According to \eqref{sigmat}, the feasibility of $h^l(\hat{\epsilon}^2) < 0$ can be established if $\hat{\epsilon}^2 - (1-\theta)^2\sigma_i^2 \geq \frac{\theta^2\left( \hat{\epsilon} + \sqrt{\hat{\epsilon}+ 4\sigma_l^2 } \right)^2}{4}$.
\end{proof}
Theorem 1 provides the conditions on the variance of the transmission delay $\sigma_i^2$ to satisfy $h^l(\hat{\epsilon}^2) < 0$ with given $\hat{\epsilon}$ and $l$. We next derive the lower bound of $l$ from Theorem 1:
\begin{corollary}
Consider a given threshold $\hat{\epsilon}$, the time slot $l$ that satisfyies $h^l(\hat{\epsilon}^2) < 0$ is bounded as 
\begin{equation}
\label{hatt}
l \geq \frac{1}{2\ln{\theta}}\ln{ \left[ \frac{\hat{\epsilon}^2 - \theta\hat{\epsilon}\sqrt{\hat{\epsilon}^2-(1-\theta)^2\sigma_i^2} - \frac{1-\theta}{1+\theta} \sigma_i^2}{\theta^2(\sigma_0^2 - \frac{1-\theta}{1+\theta} \sigma_i^2)}    \right] },
\vspace{-2pt}
\end{equation}
where $\theta$ satisfies $\hat{\epsilon}^2 - (1-\theta)^2\sigma_i^2 \geq \frac{\theta^2\left( \hat{\epsilon} + \sqrt{\hat{\epsilon}+ 4\frac{1-\theta}{1+\theta} \sigma_i^2} \right)^2}{4}$.
% \begin{equation}
% \label{feasisbility2}
%     \hat{\epsilon}^2 - (1-\theta)^2\sigma_i^2 \geq \frac{\theta^2\left( \hat{\epsilon} + \sqrt{\hat{\epsilon}+ 4\frac{1-\theta}{1+\theta} \sigma_i^2} \right)^2}{4}.
% \end{equation}
\end{corollary}

% \begin{proof}
% Based on the result in Theorem 1, we observe that if $ \hat{\epsilon}^2 - (1-\theta)^2\sigma_i^2 \geq \frac{\theta^2}{4} \min_{l} \left\{ \left( \hat{\epsilon} + \sqrt{\hat{\epsilon}+ 4\sigma_l^2 } \right)^2\right\}$, there will always be a $l$ satisfying \eqref{feasibility}. Note $\sigma_l^2 = \theta^{2l} (\sigma_0^2 - \frac{1-\theta}{1+\theta} \sigma_i^2) + \frac{1-\theta}{1+\theta} \sigma_i^2 \geq \frac{1-\theta}{1+\theta} \sigma_i^2$, we obtain the requirement in \eqref{feasisbility2}. Then, bring $\sigma_l^2$ into \eqref{sigmat}, we can obtain the lower bound in \eqref{hatt}.
% \end{proof}

From Theorem 1 and Corollary 1, we observe that both the variance of the transmission delay and the diffusion factor have a significant impact on the conditions and recovery time in our resilient design. First, given a desired threshold $\hat{\epsilon}$, \eqref{feasibility} shows that the V2V links must be designed in a way that there exists a feasible $\theta \in (0,1)$ to satisfy $h^l(\hat{\epsilon}^2) < 0$. Furthermore, the lower bound given in Corollary 1 actually captures the time interval during which the rear-end collision risk in the platoon is high after the cyber-attack. Thus, an optimal $\theta$ should be adopted to minimize the lower bound with given $\sigma_i^2$. Overall, Theorem 1 and Corollary 1 provide guidelines on how to adapt the V2V link and diffusion strategy, so as to facilitate a resilient system capable of rapidly mitigating rear-end collision risk caused by cyber-attack.

\section{Simulation Results and Analysis}
\begin{table}[t]
\scriptsize
\centering
\caption{Simulation Parameters}
\vspace{-4pt}
\label{}
    \begin{tabular}{|c|c|c|}
    \hline
    \textbf{Parameter}    & \textbf{Description}     & \textbf{Parameter}    \\ \hline
    % $l$  & Lane Width & $3.7$m \\ \hline
    $R$ & Semi circle radius & $20 \ \mathrm{m}$  \\ \hline
    $a$ & Breaking Acceleration           & $-6 \ \mathrm{m/s^2}$  \\ \hline
    $v$ & Platoon velocity           & $25 \ \mathrm{m/s}$  \\ \hline
    $d_{i-1,i}$ & Platoon headway           & $10 \ \mathrm{m}$  \\ \hline
    $t_d$ & Breaking process delay      & $0.4 \ \mathrm{s}$ \cite{1580935}  \\ \hline
    $m$ & Nakagami parameter           & $3$  \\ \hline
    $\alpha$ & Path  loss exponent     & $3.5$  \\ \hline
    $P$ & Transmission power           & $27$ dBm  \\ \hline
    $N_0$ & Noise power spectrum density          & $-174$ dBm/Hz  \\ \hline
    $B$ & Bandwidth           & $20$ MHz  \\ \hline
    $D$ & Packet size           & $3,200$ bits \cite{8778746} \\ \hline
    $\sigma_0^2$ & Initial clock offset variance & $9 \ \mathrm{s^2}$ \\ \hline
    \end{tabular}
    \vspace{-0.4cm}
\end{table}

For our simulation, we first validate the approximation of $\xi_i^{l} \sim \mathcal{N}(0,\sigma_l^2)$ and illustrate the feasible region of the proposed resilient design. Then, we present system performance of our resilient design and a traditional reliable design in dense traffic conditions with limited communication resources. Unless stated otherwise, simulation parameters are summarized in Table I and statistical results are averaged over a large number of independent runs.

\begin{figure}[t]
	\centering
	\includegraphics[scale=0.52]{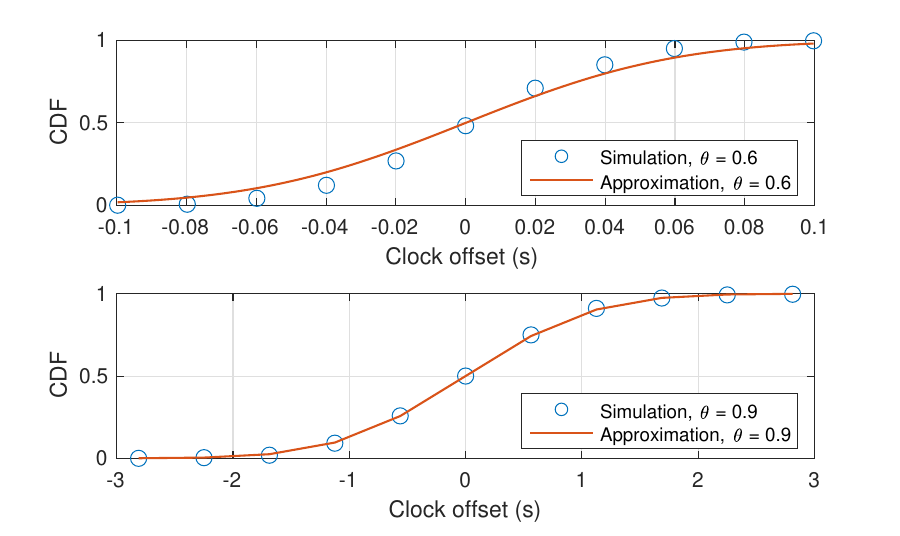}
 \vspace{-0.15cm}
	\caption{\small{Validation of approximation of $\xi_i^l$ into normal distribution.}}
     \vspace{-0.5cm}
    \label{approx}
\end{figure}

Fig. \ref{approx} shows the CDF of clock offset $\xi_i^{l}$ at round $l = 10$, compared to the CDF of normal distribution $\mathcal{N}(0,\sigma_l^2)$ with $\sigma_l^2 = \theta^{2l} \sigma_0^2 + \frac{1-\theta}{1+\theta} (1-\theta^{2l}) \sigma_i^2$. The vehicle density is $\eta = 0.01 \text{ vehicle} \mathrm{/m^2}$. As observed in Fig. \ref{approx}, the normal distribution approximation for the clock offset is reasonable, validating our following analysis and result. Fig. \ref{approx} also shows that the accuracy of the approximation will decrease with a smaller $\theta$, which aligns with the bound of distribution difference given by Berry-Esseen theorem. 
% \begin{figure}[t]
% 	\centering
% 	\includegraphics[scale=0.65]{feasible1.eps}
% 	\caption{\small{Feasible region of resilience design and reliability design with TTC requirement $T_i^c(\xi_i^{l}) \geq \hat{t} = 4s$.}\vspace{-0.5cm}}
%     \label{r-sigma}
% \end{figure}
\begin{figure}[t]
	\centering	\includegraphics[scale=0.58]{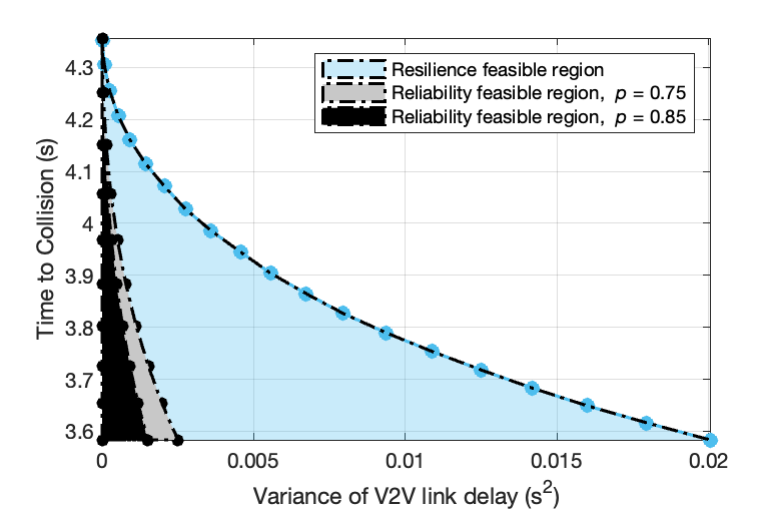}
 \vspace{-0.15cm}
	\caption{\small{Feasible region of resilient design and reliable design of $p=0.75$ and $p = 0.85$ at time slot $l = 10$.}}
    \vspace{-0.5cm}
    \label{epsilon-sigma}
\end{figure}

Next, we compare the feasible region of proposed resilient design with a traditional reliable design when the variance $\sigma_i^2$ of V2V links delay increases. In particular, vehicle $i$ under the resilient design aims at $h^l(\hat{\epsilon}^2) < 0$ by using the diffusion protocol in \eqref{Error}. Meanwhile, we consider a reliable design baseline that seeks to maintain $\mathbb{P} \left( - \hat{\epsilon} < \xi_i^{l} < \hat{\epsilon} \right) \geq p $ by simply modifying its clock after receiving clock information, i.e., $\xi_i^{l} = - u_{i-1,i}^{l-1}$. To show the impact of the clock offset on the platoon, here we transform  $- \hat{\epsilon} < \xi_i^{l} < \hat{\epsilon} $ into $T_i^c(\xi_i^{l}) \geq \hat{t}$ based on \eqref{casettc}. As observed in Fig. \ref{epsilon-sigma}, given a safety threshold $\hat{t}$, both reliable design cases can only work when the variation of V2V transmission delay is smaller than that required in our resilient design. Thus, our resilient design can work in face of large variance of transmission delay. For instance, when setting $T_i^c(\xi_i^{l}) \geq \hat{t} = 3.8$ s as the safety requirement, our resilient design is feasible for a maximum of variance of transmission delay $\sigma_i^2 = 0.009 \text{ s}^2 $, which is nine-fold more than that of the the baseline $p = 0.75$. Meanwhile, when the variance of transmission delay is fixed, our resilient design can satisfy a higher safety threshold compared to both baselines, thus, leaving more reaction time to the manual and autonomous control of vehicles and reducing the rear-end collision risks in the platoon.

\begin{figure}[t]
	\centering
	\includegraphics[scale=0.56]{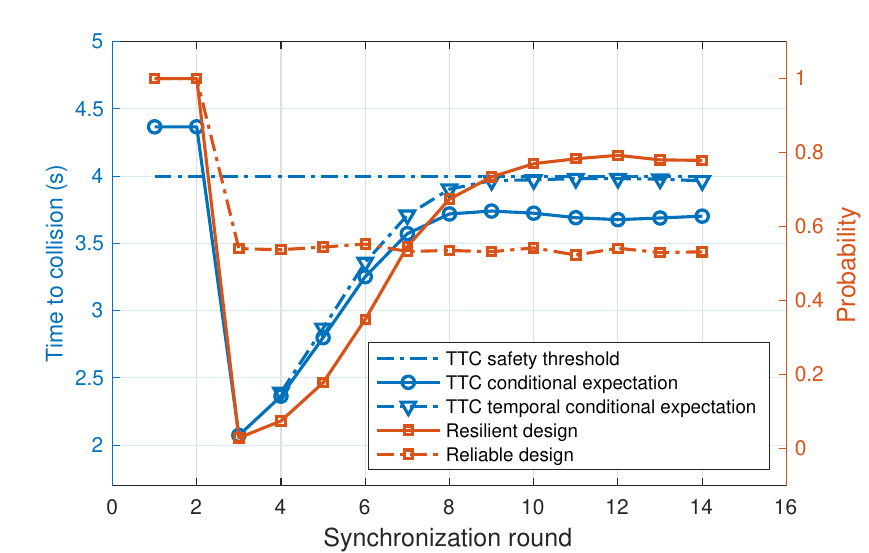}
 \vspace{-0.2cm}
	\caption{\small{ Performance after attack for both resilient and reliable design.}}
 \vspace{-0.4cm}
 \label{5}
\end{figure}
Fig. \ref{5} illustrates the effect of reliable design and resilient design after vehicle $i$ is attacked. Here we consider a dense traffic scenario $\eta = 0.03 \ \text{vehicle} \mathrm{m^2}$ with limited communication resources $B = 12 $ MHz. The diffusion factor is set as $\theta=0.45$ for our scheme and the safety requirements seek to maintain $T_i^c(\xi_i^{l}) \geq \hat{t} = 4$ s \cite{LI2020105676}. We assume the attack happens at $l=2$, and vehicle $i$ identifies it immediately and requests a re-synchronization from vehicle $i-1$. From $l=3$, vehicle $i$ starts to receive clock information from its predecessor. First, we consider the conditional expectation and temporal conditional expectation of TTC given $T_i^c(\xi_i^{l}) \leq \hat{t}$, i.e., $\mathbb{E} \left\{ T_i^c(\xi_i^{l}) \mid T_i^c(\xi_i^{l}) \leq \hat{t} \right\}$ and $\mathbb{E} \left\{ T_i^c(\xi_i^{l+1}) \mid T_i^c(\xi_i^{l}) \leq \hat{t} \right\}$, in our resilient design. As shown in Fig. \ref{5}, after $7$ rounds of recovery, our resilient design can maintain a conditional expectation and a temporal conditional expectation of TTC at $3.7$ s and $4$ s respectively. Specifically, though the conditional expectation of TTC is still $0.3$ s lower than $\hat{t} = 4$ s, our resilient design can prevent two consecutive violations of the TTC safety requirement by ensuring a temporal conditional expectation of $4$ s, i.e., $\mathbb{E} \left\{ T_i^c(\xi_i^{l+1}) \mid T_i^c(\xi_i^{l}) \leq \hat{t} \right\} = \hat{t}$. Moreover, in terms of the reliability $p = \mathbb{P} \left( T_i^c(\xi_i^{l}) \leq \hat{t} \right)$, our resilient design is better after $6$ rounds of recovery and maintain a higher reliability of $p = 0.8$, which is $45\%$ more than $p = 0.55$ in the reliable design. It should be noted that, the reliable design achieves $p = 0.55$ immediately because vehicle $i$ directly modifies its clock according to the received clock information. However, the reliable design can not improve its reliability, and only maintains a low probability to satisfy the TTC requirement in heavy traffic with limited communication resources. In contrast, our resilient design can improve its reliability by the diffusion strategy. This is because the diffusion factor $\theta$ can reduce the variance of clock offset $\xi_i^l$, i.e., $\lim_{l \rightarrow \infty} \sigma_l^2 = \frac{1-\theta}{1+\theta} \sigma_i^2 \leq \sigma_i^2$. Thus, we can tune $\theta$ for a small variance of clock offset while satisfying the condition in \eqref{feasibility}, which ensures the platoon's recovery from synchronization disruption attacks and promises a high reliability.

\section{Conclusion}
In this paper, we have proposed a resilient design against the synchronization disruption attack to vehicles platoon with the assistance of wireless V2V links. We have defined new resilience metric called TCME that uses the temporal correlation in diffusion protocol. We have derived the expression of the resilience metric and analyzed the conditions and recovery time needed in the resilient design. Simulation results validate the theoretical analysis and show that our resilient design can work in conditions of severe variance of V2V link delay in which traditional reliable design is infeasible. Moreover, after recovery from synchronization attack, the system under resilient design is able to respond to violations of clock offset threshold and obtains a higher reliability.

\ifCLASSOPTIONcaptionsoff
  \newpage
\fi



\bibliographystyle{IEEEtran}
\bibliography{bibliography}

\begin{thebibliography}{10}
\providecommand{\url}[1]{#1}
\csname url@samestyle\endcsname
\providecommand{\newblock}{\relax}
\providecommand{\bibinfo}[2]{#2}
\providecommand{\BIBentrySTDinterwordspacing}{\spaceskip=0pt\relax}
\providecommand{\BIBentryALTinterwordstretchfactor}{4}
\providecommand{\BIBentryALTinterwordspacing}{\spaceskip=\fontdimen2\font plus
\BIBentryALTinterwordstretchfactor\fontdimen3\font minus \fontdimen4\font\relax}
\providecommand{\BIBforeignlanguage}[2]{{%
\expandafter\ifx\csname l@#1\endcsname\relax
\typeout{** WARNING: IEEEtran.bst: No hyphenation pattern has been}%
\typeout{** loaded for the language `#1'. Using the pattern for}%
\typeout{** the default language instead.}%
\else
\language=\csname l@#1\endcsname
\fi
#2}}
\providecommand{\BIBdecl}{\relax}
\BIBdecl

\bibitem{8667866}
A.~Rasouli and J.~K. Tsotsos, ``Autonomous vehicles that interact with pedestrians: A survey of theory and practice,'' \emph{IEEE Transactions on Intelligent Transportation Systems}, vol.~21, no.~3, pp. 900--918, Mar. 2020.

\bibitem{9699045}
X.~Ge, Q.-L. Han, J.~Wang, and X.-M. Zhang, ``Scalable and resilient platooning control of cooperative automated vehicles,'' \emph{IEEE Transactions on Vehicular Technology}, vol.~71, no.~4, pp. 3595--3608, Apr. 2022.

\bibitem{7736181}
R.~Hult, G.~R. Campos, E.~Steinmetz, L.~Hammarstrand, P.~Falcone, and H.~Wymeersch, ``Coordination of cooperative autonomous vehicles: Toward safer and more efficient road transportation,'' \emph{IEEE Signal Processing Magazine}, vol.~33, no.~6, pp. 74--84, Nov. 2016.

\bibitem{9834918}
P.~Popovski, F.~Chiariotti, K.~Huang, A.~E. Kalør, M.~Kountouris, N.~Pappas, and B.~Soret, ``A perspective on time toward wireless 6g,'' \emph{Proceedings of the IEEE}, vol. 110, no.~8, pp. 1116--1146, Aug. 2022.

\bibitem{1580935}
S.~Biswas, R.~Tatchikou, and F.~Dion, ``Vehicle-to-vehicle wireless communication protocols for enhancing highway traffic safety,'' \emph{IEEE Communications Magazine}, vol.~44, no.~1, pp. 74--82, Jan. 2006.

\bibitem{s22176679}
R.~S. Rathore, C.~Hewage, O.~Kaiwartya, and J.~Lloret, ``In-vehicle communication cyber security: Challenges and solutions,'' \emph{Sensors}, vol.~22, no.~17, Sep. 2022.

\bibitem{10213228}
G.~Chai, W.~Wu, Q.~Yang, M.~Qin, Y.~Wu, and F.~R. Yu, ``Platoon partition and resource allocation for ultra-reliable v2x networks,'' \emph{IEEE Transactions on Vehicular Technology}, vol.~73, no.~1, pp. 147--161, Jan. 2024.

\bibitem{9046279}
D.~Zhao, H.~Qin, B.~Song, Y.~Zhang, X.~Du, and M.~Guizani, ``{A Reinforcement Learning Method for Joint Mode Selection and Power Adaptation in the V2V Communication Network in 5G},'' \emph{IEEE Transactions on Cognitive Communications and Networking}, vol.~6, no.~2, pp. 452--463, Jun. 2020.

\bibitem{9741813}
G.~Ding, J.~Yuan, G.~Yu, and Y.~Jiang, ``Two-timescale resource management for ultrareliable and low-latency vehicular communications,'' \emph{IEEE Transactions on Communications}, vol.~70, no.~5, pp. 3282--3294, May. 2022.

\bibitem{9013252}
T.~Zeng, O.~Semiari, W.~Saad, and M.~Bennis, ``Dependence control for reliability optimization in vehicular networks,'' in \emph{Proc. IEEE Global Communications Conference (GLOBECOM)}, Waikoloa, HI, USA, Dec. 2019, pp. 1--6.

\bibitem{1566581}
Q.~Li and D.~Rus, ``Global clock synchronization in sensor networks,'' \emph{IEEE Transactions on Computers}, vol.~55, no.~2, pp. 214--226, Feb. 2006.

\bibitem{8778746}
T.~Zeng, O.~Semiari, W.~Saad, and M.~Bennis, ``Joint communication and control for wireless autonomous vehicular platoon systems,'' \emph{IEEE Transactions on Communications}, vol.~67, no.~11, pp. 7907--7922, Nov. 2019.

\bibitem{546270}
H.~Raza and P.~Ioannou, ``Vehicle following control design for automated highway systems,'' \emph{IEEE Control Systems Magazine}, vol.~16, no.~6, pp. 43--60, Dec. 1996.

\bibitem{LI2020105676}
Y.~Li, D.~Wu, J.~Lee, M.~Yang, and Y.~Shi, ``Analysis of the transition condition of rear-end collisions using time-to-collision index and vehicle trajectory data,'' \emph{Accident Analysis \& Prevention}, vol. 144, p. 105676, Sep. 2020.

\bibitem{bryson1974heavy}
M.~C. Bryson, ``Heavy-tailed distributions: properties and tests,'' \emph{Technometrics}, vol.~16, no.~1, pp. 61--68, Feb. 1974.

\bibitem{berry1941accuracy}
A.~C. Berry, ``The accuracy of the gaussian approximation to the sum of independent variates,'' \emph{Transactions of the american mathematical society}, vol.~49, no.~1, pp. 122--136, Jan. 1941.

\bibitem{itō1974diffusion}
K.~It{\=o} and H.~McKean, \emph{Diffusion Processes and Their Sample Paths}, ser. Die Grundlehren der mathematischen Wissenschaften in Einzeldarstellungen.\hskip 1em plus 0.5em minus 0.4em\relax Springer-Verlag, 1974.

\end{thebibliography}

%\begin{thebibliography}{1}

%\end{thebibliography}


% \begin{IEEEbiography}{}

% \end{IEEEbiography}

% \begin{IEEEbiographynophoto}{}


%\newpage




%\vfill


\end{document}